%% file: secure_coding_v5.tex
\documentclass[10pt,conference]{IEEEtran}
\usepackage{graphicx}
\usepackage{times}
\usepackage{amsmath}
\usepackage{subfigure}
\usepackage{longtable}
\usepackage{amssymb}
\usepackage{epsfig,latexsym,subfigure,cite}

\IEEEoverridecommandlockouts

\include{symboldef}

\begin{document}


\title{Secure Nested Codes for Type II Wiretap Channels}

\author{
\authorblockN{Ruoheng Liu, Yingbin Liang, and H. Vincent Poor}
\authorblockA{Department of Electrical Engineering, Princeton University\\
Princeton, NJ 08544 \\
email: \{rliu,yingbinl,poor\}@princeton.edu}
\and
\authorblockN{Predrag~Spasojevi\'{c}}
\authorblockA{WINLAB, ECE, Rutgers University \\
North Brunswick, NJ 08902\\
email: spasojev@winlab.rutgers.edu}
\thanks{This research was supported by the National Science Foundation under Grants ANI-03-38807 and CNS-06-25637.}
}

\maketitle

\begin{abstract}

This paper considers the problem of secure coding design for a type II wiretap
channel, where the main channel is noiseless and the eavesdropper channel is a
general binary-input symmetric-output memoryless channel. The proposed secure
error-correcting code has a \emph{nested code} structure. Two secure nested
coding schemes are studied for a type II Gaussian wiretap channel. The nesting
is based on cosets of a \emph{good code} sequence for the first scheme and on
cosets of the dual of a good code sequence for the second scheme. In each case,
the corresponding achievable rate-equivocation pair is derived based on the
threshold behavior of good code sequences. The two secure coding schemes
together establish an achievable rate-equivocation region, which almost covers
the secrecy capacity-equivocation region in this case study. The proposed
secure coding scheme is extended to a type II binary symmetric wiretap channel.
A new achievable perfect secrecy rate, which improves upon the previously
reported result by Thangaraj \emph{et al.}, is derived for this channel.
\end{abstract}

\section{Introduction}

Fostered by the rapid proliferation of wireless communication devices,
technologies, and applications, the need for reliable and secure data
communication over wireless networks is more important than ever before. Due to
its broadcast nature, wireless communication is particularly susceptible to
eavesdropping. Security and privacy systems have become critical for wireless
providers and enterprise networks. The aim of this paper is to study practical
secure coding schemes for wireless communication systems.


\begin{figure}[hbt]
 \centerline{\includegraphics[width=0.8\linewidth,draft=false]{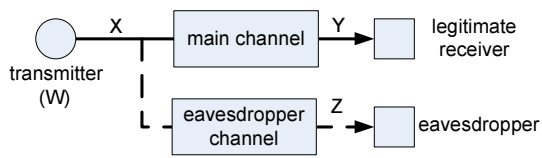}}
  \caption{\small Wiretap channel model}
  \label{fig:wtch}
\end{figure}
Shannon provided the first truly scientific treatment of secrecy in
\cite{Shannon:BSTJ:49}, where a secret key is considered to protect
confidential messages. The ingenuity of his remarkable work lies not only in
the method used therein but also in the incisive formulation that Shannon made
of the secrecy problem based on information-theoretic concepts. Later, Wyner
proposed an alternative approach to secure communication schemes in his seminal
paper \cite{Wyner:BSTJ:75}, where he introduced the so-called {\it wiretap
channel} model. As shown in Fig.~\ref{fig:wtch}, the confidential communication
via a discrete, memoryless main channel is eavesdropped upon by a wiretapper,
who has access to the degraded channel output.
Wyner demonstrated that secure communication is possible without sharing a
secret key and determined the secrecy capacity for a wiretap channel.
Construction of explicit and practical secure encoders and decoders whose
performance is as good as promised by Wyner is still an unsolved problem in the
general case, except for the binary erasure wiretap channel
\cite{Ozarow:BSTJ:84,Thangaraj:ARXIV:05,Wei:IT:91}.

We note that channel coding and secrecy coding are closely related. Roughly
speaking, the goal of channel coding is to send a message with sufficient
redundancy so that it can be understood by the receiver; whereas the goal of
secrecy coding is to provide sufficient randomness so that the message can not
be understood by anyone else. In modern communication networks,
error-correcting codes have traditionally been designed to ensure communication
reliability. Various coding techniques have been thoroughly developed and
tested for ensuring reliability of virtually all current single-user,
point-to-point physical channels. However, only very limited work has
considered ways of using error-correcting codes to also ensure security. In
\cite{Ozarow:BSTJ:84}, Ozarow and Wyner considered error-correcting code design
for a type II binary erasure wiretap channel based on a coset coding scheme.
More recently, low-density parity-check (LDPC) based coding design has been
studied for binary erasure wiretap channels in~\cite{Thangaraj:ARXIV:05}, where
the authors have also presented code constructions for a type II binary
symmetric wiretap channel based on error-detection codes. In another line of
recent related work,
secret key agreement protocols based on powerful LDPC codes have been studied
by several authors \cite{Bloch:ISIT:06,Muramatsu:TFECCS:06,Barros:ITA:07}.
Designing practical secure coding schemes for additive white Gaussian noise
(AWGN) wiretap channels, for example, is still an open problem.

In this work, we focus on secure coding schemes for a type II wiretap channel,
where the main channel is noiseless and the eavesdropper channel is a
binary-input symmetric-output memoryless (BISOM) channel. We first review and
summarize the prior results of
\cite{Wyner:BSTJ:75,Ozarow:BSTJ:84,Thangaraj:ARXIV:05}. Inspired
by~\cite{Zamir:IT:02}, we propose a more general \emph{secure nested code}
structure. Next, we consider a type II AWGN wiretap channel and describe two
secure coding schemes, both of which have a nested structure. The nesting is
based on cosets of a \emph{good code} sequence for the first scheme and on
cosets of the dual of a good code sequence for the second scheme. In each case,
we derive the corresponding achievable rate-equivocation pair based on the
threshold behavior of good code sequences \cite{mack,rpe-good-IT}.
By combining the two secure coding schemes, we establish an achievable
rate-equivocation region, which almost covers the secrecy capacity-equivocation
region for the described case study. Finally, we extend the secure coding to a
type II binary symmetric wiretap channel and derive a new achievable (perfect)
secrecy rate, which improves upon the result previously reported in
\cite{Thangaraj:ARXIV:05}.

\section{Preliminaries}

We review here some definitions and results from
\cite{Wyner:BSTJ:75,Ozarow:BSTJ:84,Thangaraj:ARXIV:05} and propose a secure
nested coding structure, which serves as preliminary material for the rest of
the paper.

\subsection{General Wiretap Channel Model}

We consider the classic wiretap channel \cite{Wyner:BSTJ:75} illustrated in
Fig.~\ref{fig:wtch}, where the transmitter sends a confidential message to a
legitimate receiver via the main channel in the presence of an eavesdropper,
who listens to the message through its own channel. Both the main and the
eavesdropper channels are discrete memoryless, and in particular, the
eavesdropper channel is a degraded version of the main channel. A confidential
message $w \in \Wc$ is mapped into a channel input sequence $\xv=[x_1, x_2,
\dots, x_n]$ of length $n$, where $\Wc=\{1,\dots,M\}$ and $M$ is the number of
distinct confidential messages that may be transmitted. The outputs from the
main channel and the eavesdropper channel are $\yv$ and $\zv$, respectively.
The level of ignorance of the eavesdropper with respect to the confidential
message is measured by the equivocation $H(W|\Zv)$. A rate-equivocation pair
$(R,R_e)$ is \emph{achievable} if there exists a rate $R$ code sequence with
the average probability of error $P_e\rightarrow 0$ as the code length $n$ goes
to infinity and with the equivocation rate $R_e$ satisfying
\[R_e\le \lim_{n\rightarrow \infty} H(W|\Zv)/n.\]
Perfect secrecy requires that, for any $\epsilon_0>0$ there exists a
sufficiently large n so that the normalized equivocation satisfies
\[H(W|\Zv)/n \ge H(W)/n-\epsilon_0 .\]
Hence, perfect secrecy happens when $R_e=R$, i.e., all the information
transmitted over the main channel is secret. The capacity-equivocation region
of the wiretap channel $X\rightarrow (Y,Z)$ \cite{Wyner:BSTJ:75} contains
rate-equivocation pairs $(R,R_e)$ that satisfy
\begin{align}
R_e\le R&\le \max_{p(x)}I(X;Y) \notag\\
0\le R_e&\le \max_{p(x)} [I(X;Y)-I(X;Z)].
\end{align}

\subsection{Wyner Codes and Secrecy Bins}

It is instructive to review first the problem of \emph{unstructured} secure
code design in terms of the stochastic encoding scheme introduced by Wyner
\cite{Wyner:BSTJ:75}. As demonstrated in~\cite{Wyner:BSTJ:75} the secrecy
capacity of the wiretap channel is achieved by using a stochastic encoder,
where a {\em mother} codebook $\C_0(n)$ of length $n$ is randomly partitioned
into ``secret bins'' or sub-codes $\{\C_1(n),C_2(n),\dots,\C_M(n)\}$. A message
$w$ is associated with a sub-code $\C_w(n)$ and the transmitted codeword is
randomly selected within the sub-code. Such codebook allows for decomposing the
twofold objective of achieving both reliability and secrecy into two separate
objectives. The mother code $\C_0(n)$ provides enough redundancy so that the
legitimate receiver can decode the message reliably, whereas each sub-code is
sufficiently large and, hence, introduces enough randomness so that the
eavesdropper's uncertainty about the transmitted message can be guaranteed.

Even though \cite{Wyner:BSTJ:75} does not describe a structured coding scheme,
it does suggest that encoding for reliability and confidentiality would be to
partition the mother code into sub-codes. This idea has been extended to
structured or semi-structured codes by using coset codes in
\cite{Ozarow:BSTJ:84,Thangaraj:ARXIV:05}.

\subsection{Secure Nested Codes}

In the following, we construct secure error-correcting codes with the {\it
nested code} structure \cite{Zamir:IT:02}.\footnote{In this paper, we consider
binary-input wiretap channels and nested linear codes. This idea can be
extended to nested lattice codes for channels with continuous inputs.}

We consider a nested linear code pair $(\C_0(n),\, \C_1(n))$, where $\C_0(n)$
is a {\it fine} code of rate $R_0$, and $\C_1(n)$ a {\it coarse} code of rate
$R_1$. We use the fine code $\C_0(n)$ as the mother code, which is partitioned
into $M$ sub-codes consisting of the coarse code $\C_1(n)$ and its cosets. Each
coset corresponds to a confidential message. The transmitter encodes a message
$w\in\Wc$ into an $n$-tuple of coded symbols randomly selected within the
corresponding coset $\C_w(n)$. By determining the coset of the transmitted
codeword, the legitimate receiver can retrieve the confidential message $w$.
The redundancies provided by each coset are used to confuse the eavesdropper
who has full knowledge about the code and its cosets. We refer to a code
structured in this manner as a \emph{secure nested code}. We note that the code
$\C_1(n)$ and its cosets have the same (Hamming) distance properties. Hence,
the secure coding design problem is to find a suitable nested code pair
$(\C_0(n),\, \C_1(n))$ that satisfies both confidentiality and reliability
requirements. Denote by $\{\C(n)\}$ a sequence of binary linear codes, where
$\C(n)$ is an $(n,\,k_n)$ code having a common rate $R_c=k_n/n$. Now, we define
the secure nested code sequence as follows.
\begin{definition} [{\bf secure code sequence}]
$\{\C_0(n),\,\C_1(n)\}$ is a secure nested code sequence if $\C_0(n)$ is a
(mother) fine code of rate $R_0$, and $\C_1(n)$ is a coarse code of rate $R_1$
so that $\C_1(n)\subseteq \C_0(n)$ and $R_1\le R_0$. The information rate of
this code sequence is $R_0-R_1$.
\end{definition}

\subsection{Good Code and Its Noise Threshold}

Following MacKay~\cite{mack}, we say that a code sequence $\{\C(n)\}$ is
\emph{good} if it achieves arbitrarily small word (bit) error probability when
transmitted over a \emph{noisy} channel at a nonzero rate $R_c$.
\emph{Capacity-achieving} codes are good codes whose rate $R_c$ is equal to the
channel capacity. The class of good codes includes, for example, turbo, LDPC,
and repeat-accumulate codes, whose performance is characterized by a
\emph{threshold} behavior in a single channel model \cite{rpe-good-IT}.
\begin{definition} [{\bf noise threshold}]
For a (single) channel model described by a single parameter, the noise
threshold of a code sequence $\{\C(n)\}$ is defined as the worst case channel
parameter value at which the word (bit) error probability decays to zero as the
codeword length $n$ increases.
\end{definition}
For example, the noise threshold is described in terms of the erasure rate
threshold $\delta^{\star}$ for a binary erasure channel (BEC) and the SNR
threshold $\lambda^{\star}$ for a binary-input AWGN (BI-AWGN) channel. Noise
thresholds associated with good codes and the corresponding maximum-likelihood
(ML), ``typical pair'', and iterative decoding algorithms have been studied in
\cite{sham:saso:IT:02,rich:urba-ldpc:IT,jin::mcel:ISITA:00}.

\subsection{Type II Wiretap Channel}

The type II wiretap channel was introduced by Ozarow and Wyner in
\cite{Ozarow:BSTJ:84} as a special binary-input wiretap channel with a
noiseless main channel.
Throughout the paper, we focus on type II wiretap channels associated with
different eavesdropper channels.
\begin{example}[{\rm BEC-WT}]
Let BEC-WT($\epsilon$) denote a binary-input wiretap channel where the main
channel is noiseless and the eavesdropper channel is a BEC with erasure rate
$\epsilon$. We refer to such a channel as the type II binary erasure wiretap
channel. The secrecy capacity of BEC-WT($\epsilon$), $C_{s,{\rm
BEC}}(\epsilon)$, equals $\epsilon$.
\end{example}

Let $\{\C^{\bot}(n)\}$ be a sequence of dual codes, where
\begin{align*}
\C^{\bot}(n)=\{\xv\in \{0,1\}^{n}\,|\,
 \xv\cdot\yv=0,~\forall~\yv\in\C(n)\}
\end{align*}
is the dual code of $\C(n)$. By employing the dual code as the coarse code in
the secure nested code structure, we reorganize the results of
\cite{Ozarow:BSTJ:84,Thangaraj:ARXIV:05} in the following lemma.


\begin{lemma} \label{lem:bec}
Consider a sequence of binary linear codes  $\{\C(n)\}$ of rate $R_c$ and
erasure rate threshold $\delta^{\star} \le 1-R_c$ (for the BEC). Let
\begin{align}
\C_0(n)=\{0,1\}^n\quad \text{and} \quad \C_1(n)=\C^{\bot}(n).
\end{align}
Suppose that the secure nested code sequence $\{\C_0(n),\,\C_1(n)\}$ is
transmitted over a BEC-WT($\epsilon$). Then, if
\begin{align}
\epsilon\ge 1-\delta^{\star}, \label{eq:becc}
\end{align}
the achievable rate-equivocation pair $(R,R_e)=(R_c,\,R_c)$.
\end{lemma}
Lemma~\ref{lem:bec} illustrates that one can design practical secure codes to
achieve perfect secrecy with a certain transmission rate (below the secrecy
capacity) for a BEC-WT. The condition (\ref{eq:becc}) implies that to achieve
the secrecy capacity, the coding scheme requires a capacity-achieving code
sequence as the dual code of the coarse code.

Two capacity-achieving LDPC code sequences for BECs have been described in
\cite{Oswald:IT:02}, called the \emph{Tornado sequence} $\{\C_{\rm T}(n)\}$ and
the \emph{right-regular sequence} $\{\C_{\rm R}(n)\}$. For both of these
sequences, the erasure rate threshold $\delta^{\star}=1-R_c=\epsilon$.
\begin{corollary}
Consider LDPC code sequences $\{\C_{\rm T}(n)\}$ and $\{\C_{\rm R}(n)\}$ of
rate $R_c$. Let $\C_0(n)=\{0,1\}^n$ and
\begin{align}
\C_1(n)=\C_{T}^{\bot}(n) \quad\text{or} \quad
\C_1(n)=\C_{R}^{\bot}(n).
\end{align}
Then, the secure nested code sequence $\{\C_0(n),\,\C_1(n)\}$ achieves the
secrecy capacity of BEC-WT($1-R_c$).
\end{corollary}

\section{Main Results}
In this section, we consider practical coding design for secure communication
over a type II AWGN wiretap channel.
\begin{figure}[hbt]
 \centerline{\includegraphics[width=0.95\linewidth,draft=false]{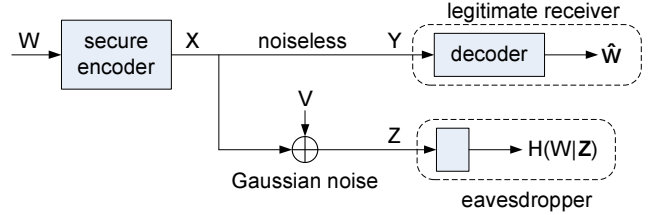}}
  \caption{\small  Type II AWGN wiretap channel}
  \label{fig:awgnwt}
\end{figure}
As shown in Fig.~\ref{fig:awgnwt}, the eavesdropper channel is a BI-AWGN
channel characterized by transition probabilities
\begin{align}
g(z|X=1)&=\frac{1}{\sqrt{2\pi}}\exp\left[\frac{-(z+\sqrt{2\lambda})^2}{2}\right] \notag \\
\text{and} \qquad
g(z|X=-1)&=\frac{1}{\sqrt{2\pi}}\exp\left[\frac{-(z-\sqrt{2\lambda})^2}{2}\right]
\label{eq:gfunc}
\end{align}
where $\lambda =E_s/{N_0}$ is the ratio of the energy per coded symbol to the
one-sided spectral noise density, which is referred to as the SNR of the
eavesdropper channel. We denote this channel with AWGN-WT($\lambda$). The
capacity-equivocation region of AWGN-WT($\lambda$) contains rate-equivocation
pairs $(R,R_e)$ that satisfy
\begin{align}
R_e\le R&\le  1 \notag\\
0\le R_e&\le  1-C_{\rm BI-AWGN}(\lambda) \label{eq:awgn-re}
\end{align}
where
\begin{align}
C_{\rm BI-AWGN}&(\lambda)=1-\notag\\
&
\frac{1}{\sqrt{\pi}}\int_{-\infty}^{+\infty}e^{-(y-\sqrt{\lambda})^2}
\log_2\bigl(1+e^{-4y\sqrt{\lambda}}\bigr)\,dy \label{eq:C-awgn}
\end{align}
is the channel capacity of BI-AWGN channel with SNR $\lambda$.

In the following, we consider two approaches to designing secure codes, both of
which have a nested structure.  In each case, we derive the corresponding
achievable rate-equivocation pair based on the threshold behavior of good
codes~\cite{rpe-good-IT}.

We note that even for a general BI-AWGN channel (without a secrecy constraint),
designing practical capacity-achieving codes is still an open problem. Hence,
to allow secure codes to be implementable, we either loosen the perfect secrecy
requirement (allow for a nonzero gap between the transmission rate and the
equivalent rate) or reduce the transmission rate compared with the capacity. In
the first approach, we construct practical codes ensuring an equivocation rate
that is below the transmission rate; whereas, in the second approach, we design
secure codes to achieve perfect secrecy with a transmission rate that is below
the secrecy capacity. We summarize code designs and the corresponding
achievable rate-equivocation pair as follows.

\subsection{Approach I: Good Coarse Code}

In Approach I, we use a good code as the coarse code $\C_1(n)$.

\begin{theorem} \label{th:awgn-1}
Consider a sequence of secure nested codes $\{\C_0(n),\,\C_1(n)\}$, where
$\C_0(n)=\{0,1\}^n$ and $\{\C_1(n)\}$ is a good binary linear code sequence of
rate $R_1$ and SNR threshold $\lambda^{\star}$ (for BI-AWGN channels). Suppose
that the secure nested code sequence $\{\C_0(n),\,\C_1(n)\}$ is transmitted
over AWGN-WT($\lambda$). Then, if $\lambda\ge \lambda^{\star}$, the
rate-equivocation pair
\begin{align}
(R,R_e)=\bigl(1-R_1,\,1-C_{\rm BI-AWGN}(\lambda)\bigr)
\end{align}
is achievable.
\end{theorem}
Theorem~\ref{th:awgn-1} is proved in Appendix~\ref{sec:pft1}. Note that if the
code sequence $\C_1(n)$ is not a capacity-achieving sequence, then
$$R_1<C_{\rm BI-AWGN}(\lambda^{\star})\le C_{\rm BI-AWGN}(\lambda).$$
The gap between the rate $R_1$ and the capacity $C_{\rm
BI-AWGN}(\lambda^{\star})$ implies $R_e\le R$. Hence this approach cannot
achieve perfect secrecy when using non capacity-achieving sequences.

\begin{example} \label{ex:ldpc36}
Consider a sequence of $(4,6)$ regular LDPC codes $\{C_{\rm LDPC}(n)\}$
\cite{gallager_LDPC}. Let
\[\C_0(n)=\{0,1\}^n\quad \text{and} \quad \C_1(n)=\C_{\rm LDPC}(n).\]
The design rate of $\C_{\rm LDPC}(n)$ is  $R_1=1/3$. The SNR threshold of
$\{\C_{\rm LDPC}(n)\}$ satisfies $\lambda^{\star}\le 0.302$
under typical pair decoding \cite{jin::mcel:ISITA:00} (and hence, under ML
decoding). Assume that the secure code $\{\C_0(n),\,\C_1(n)\}$ is transmitted
over AWGN-WT($\lambda=0.302$). The achievable rate-equivocation pair is given
by
\begin{align*}
(R,R_e)&=\bigl(1-R_1,\,1-C_{\rm BI-AWGN}(0.302)\bigr)\\
&=(2/3,\,0.663).
\end{align*}
In this case, the gap between the transmission rate and the equivalent rate is
less then $0.004$.
\end{example}
Moreover, Approach I can be extended to the general AWGN wiretap channel (the
main channel is also a BI-AWGN channel) by constructing a nested LDPC code
pair.

\subsection{Approach II: Dual Good Code as Coarse Code}

In Approach II, we use the dual code of a good code as the coarse code
$\C_1(n)$. Let
\begin{align}
Q(x)=\int_{x}^{\infty}\frac{1}{\sqrt{2\pi}}\exp\left(-\frac{z^2}{2}\right)\,dz.
\end{align}

\begin{theorem} \label{th:awgn-2}
Consider a sequence of good binary linear codes  $\{\C(n)\}$ of rate $R_c$ and
erasure rate threshold $\delta^{\star}$ (for BECs). Let
\begin{align}
\C_0(n)=\{0,1\}^n\quad \text{and} \quad \C_1(n)=\C^{\bot}(n).
\end{align}
Suppose that the secure nested code sequence $\{\C_0(n),\,\C_1(n)\}$ is
transmitted over an AWGN-WT($\lambda$). Then, if
\begin{align}
Q(\sqrt{2\lambda})\ge (1-\delta^{\star})/2, \label{eq:con-awgn}
\end{align}
the rate-equivocation pair $(R,R_e)=(R_c,\,R_c)$ is achievable.
\end{theorem}
We provide the proof in Appendix~\ref{sec:pft2}. Theorem~\ref{th:awgn-2}
illustrates that $R_e=R$ if the eavesdropper channel SNR $\lambda$ satisfies
the condition (\ref{eq:con-awgn}). Hence we can achieve perfect secrecy without
using capacity-achieving codes.

\begin{example}
We use a sequence of $(4,6)$ regular LDPC codes $\{C_{\rm LDPC}(n)\}$ of rate
$R_c=1/3$ as in Example~\ref{ex:ldpc36}. Let
\[\C_0(n)=\{0,1\}^n\quad \text{and} \quad \C_1(n)=\C^{\bot}_{\rm LDPC}(n).\]
The erasure rate threshold of $\{\C_{\rm LDPC}(n)\}$ is lower-bounded as
$\delta^{\star}\ge 0.665$ under typical pair decoding \cite{khandekar:phd} (and
hence, under ML decoding). Assume that the secure code sequence
$\{\C_0(n),\,\C_1(n)\}$ is transmitted over AWGN-WT($\lambda \ge 0.465$). Since
$$2Q(\sqrt{2\lambda})\ge 0.335\ge 1-\delta^{\star},$$
Theorem~\ref{th:awgn-2} implies that this code sequence can achieve perfect
secrecy at the transmission rate $1/3$.
\end{example}

\begin{corollary} \label{cor:awgn2}
Consider LDPC code sequences $\{\C_{\rm T}(n)\}$ and $\{\C_{\rm R}(n)\}$. Let
$\C_0(n)=\{0,1\}^n$ and
\begin{align}
\C_1(n)=\C_{\rm T}^{\bot}(n) \quad\text{or} \quad \C_1(n)=\C_{\rm R}^{\bot}(n).
\end{align}
Assume that the nested code sequence $\{\C_0(n),\,\C_1(n)\}$ is transmitted
over AWGN-WT($\lambda$). Then, perfect secrecy can be achieved at (and below)
the transmission rate $2Q(\sqrt{2\lambda})$.
\end{corollary}

Corollary~\ref{cor:awgn2} implies that the gap between the secrecy capacity
(\ref{eq:awgn-re}) and the achievable (perfect) secrecy rate is
\begin{align*}
\Delta=1-C_{\rm BI-AWGN}(\lambda)-2Q(\sqrt{2\lambda}).
\end{align*}
The gap $\Delta$ can be reduced if one can find a tighter sufficient condition
than (\ref{eq:con-awgn}).

\subsection{Achievable Rate-Equivocation Region}
Now, we consider the achievable rate-equivocation region based on practical
codes for AWGN-WT($\lambda$). For a given channel SNR $\lambda$, we choose a
good code sequence $\{\C(n)\}$ of rate $R^{\star}$ so that its SNR threshold
$\lambda^{\star}\le\lambda$. Let $\C_0(n)=\{0,1\}^{n}$ and select $\C_1(n)$
from
$$\{0\}^{n}, ~ \{0,1\}^{n}, ~\C(n),~ \text{and} ~\C_{\rm R}^{\bot}(n)$$
corresponding to different equivocation rate requirements. By using a
time-sharing strategy, we can show that the secure coding scheme achieves the
rate-equivocation region
\begin{align}
\Rs_{\rm AWT}= \text{convex hull} \left\{
                                        \begin{array}{c}
                                          (0,0),\\
                                          \bigl(2Q(\sqrt{2\lambda}),2Q(\sqrt{2\lambda})\bigl), \\
                                          \bigl(1-R_1,1-C_{\rm BI-AWGN}(\lambda)\bigr), \\
                                          \bigl(1,1-C_{\rm BI-AWGN}(\lambda)\bigr),  ~(1,0)
                                        \end{array}
                                      \right\}. \notag
\end{align}

\begin{example} \label{ex:GSER}
\begin{figure}[hbt]
\centerline{\includegraphics[width=0.9\linewidth,draft=false]{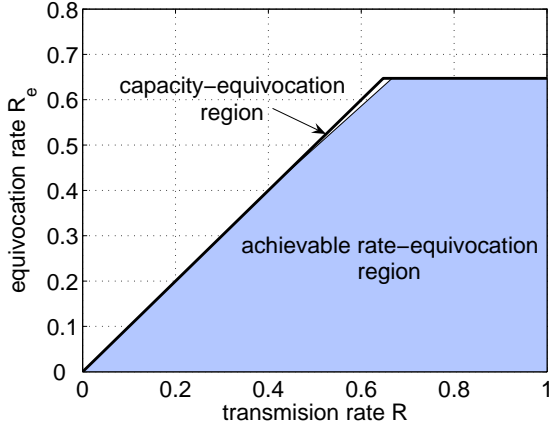}}
  \caption{\small Achievable rate-equivocation region for AWGN-WT($0.32$)}
  \label{fig:ach-region}
\end{figure}
Consider AWGN-WT($\lambda=0.32$) and the good code sequence $\{\C(n)\}=\{C_{\rm
LDPC}(n)\}$ described in Example~\ref{ex:ldpc36}, whose SNR threshold is
bounded as
$$\lambda^{\star}\le 0.302 <0.32=\lambda.$$
Fig.~\ref{fig:ach-region}. depicts the region $\Rs_{\rm AWT}$ and compares it
with the capacity-equivocation region for AWGN-WT($0.32$).
\end{example}

\section{Type II Binary Symmetric Wiretap Channel}

In this section, we study the type II binary symmetric wiretap channel. This
channel was studied previously in \cite{Thangaraj:ARXIV:05} and an achievable
secrecy rate based on error-detecting codes was given. In the following, we
apply the coding technique in Approach II and obtain an improved secrecy rate
with respect to the result in \cite{Thangaraj:ARXIV:05}.

Let BSC-WT($q$) be a type II binary symmetric wiretap channel, where the
eavesdropper channel is a binary symmetric channel (BSC) with crossover rate
$q$. The secrecy capacity of BSC-WT($q$) $C_{s,{\rm BSC}}(q)=h(q)$, where
$h(q)$ is a binary entropy function. We first summarize the result of
\cite{Thangaraj:ARXIV:05} in the following lemma.

\begin{lemma} \label{lem:bsc-old}
Consider a sequence of error-detecting codes $\{\C_{\rm D}(n)\}$ of rate $R_1$,
whose detection error rate is less than $2^{-nR_1}$. Let $\C_0(n)=\{0,1\}^n$
and $\C_1(n)=\C_{\rm D}(n).$ Assume that the nested code sequence
$\{\C_0(n),\,\C_1(n)\}$ is transmitted over a BSC-WT($q$). The maximum possible
secrecy rate that can be achieved by this construction is $-\log_2(1-q)$.
\end{lemma}
The authors of \cite{Thangaraj:ARXIV:05} have also stated that error-detecting
codes include Hamming codes and double-error-correcting BCH codes; however,
most known classes of error-detecting codes have $R_1=0$. Hence, the
implementation of such secure codes described in Lemma~\ref{lem:bsc-old} is
still an open problem.

Following Approach II, we construct implementable perfect secrecy nested codes
for BSC-WT($q$) as follows.
\begin{theorem} \label{th:bsc-new}
Consider a sequence of good binary linear codes  $\{\C(n)\}$ of rate $R_c$ and
erasure rate threshold $\delta^{\star}$ (for BECs). Let $\C_0(n)=\{0,1\}^n$ and
$\C_1(n)=\C^{\bot}(n). $ Suppose that the secure nested code sequence
$\{\C_0(n),\,\C_1(n)\}$ is transmitted over an BSC-WT($q$). Then, if
\begin{align}
q\ge (1-\delta^{\star})/2, \label{eq:con-bsc}
\end{align}
the rate-equivocation pair $(R,R_e)=(R_c,\,R_c)$ is achievable.
\end{theorem}
\begin{proof}
The proof is similar to the one described in Appendix~\ref{sec:pft2} by
constructing an equivalent BSC channel as in Fig.~\ref{fig:eqbsc}.
\end{proof}
\begin{figure}[hbt]
 \centerline{\includegraphics[width=0.6\linewidth,draft=false]{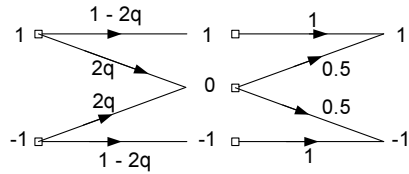}}
  \caption{\small Equivalent BSC channel with crossover probability $q$ }
  \label{fig:eqbsc}
\end{figure}
By using the LDPC code sequence $\{\C_{\rm R}(n)\}$, i.e., setting
$\C_1(n)=\C_{\rm R}^{\bot}(n)$, the achievable (perfect) secrecy rate under
this construction is $2q$, which is better than $-\log_2(1-q)$ derived in
\cite{Thangaraj:ARXIV:05}. We compare the achievable (perfect) secrecy rate
with the secrecy capacity for BSC-WT($q$) in Fig.~\ref{fig:Absc}.
\begin{figure}[hbt]
 \centerline{\includegraphics[width=0.9\linewidth,draft=false]{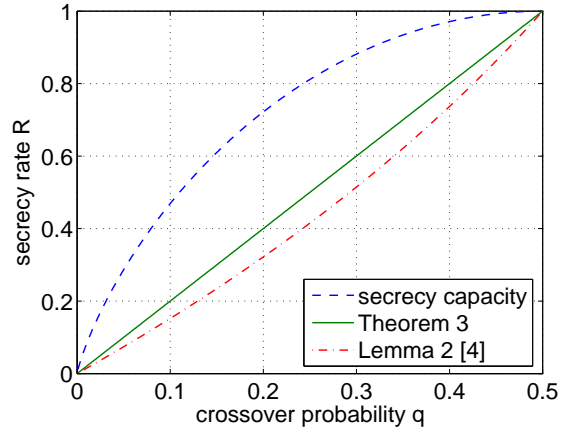}}
  \caption{\small Achievable secrecy rate vs secrecy capacity for BSC-WT($q$)}
  \label{fig:Absc}
\end{figure}

\section{Conclusion}

In this paper, we have addressed the problem of secure coding design for a type
II wiretap channel. A secure error-correcting code has been proposed in terms
of a \emph{nested code} structure. Two secure nested coding schemes have been
studied for a type II AWGN wiretap and the corresponding achievable
rate-equivocation pair has been derived based on the threshold behavior of good
code sequences. Combining the two secure coding schemes, we have established an
achievable rate-equivocation region, which almost covers the secrecy
capacity-equivocation region in this case study. Furthermore, we have also
applied the proposed secure coding scheme to a type II binary symmetric wiretap
channel, and have obtained a new achievable (perfect) secrecy rate, which
improves upon the previous result of \cite{Thangaraj:ARXIV:05}.

\appendix

\subsection{Proof of Theorem~\ref{th:awgn-1}} \label{sec:pft1}

The reliability at the desired receiver can be ensured since the main channel
is noiseless. Now, we calculate only the equivocation:
\begin{align}
H(W|\Zv)&=H(W,\Zv)-H(\Zv) \notag\\
&=H(W,\Xv,\Zv)-H(\Xv|W,\Zv)-H(\Zv)\notag\\
&\ge H(\Xv)-H(\Xv|W,\Zv)-I(\Xv;\Zv)\notag\\
&\ge n-H(\Xv|W,\Zv)-nC_{\rm BI-AWGN}(\lambda).\label{eq:equiv}
\end{align}

In order to calculate the conditional entropy $H(\Xv| W,\Zv)$, we consider the
following situation. Let us fix $W=w$ and assume that the transmitter sends a
codeword $\xv\in\C_w(n)$. Given index $W=w$, the eavesdropper decodes the
codeword $\xv$ based on the received sequence $\zv$. Let $P(\C_w,n)$ denote the
average probability of error under ML decoding at the eavesdropper incurred by
using coset $\C_w(n)$. We note that the code $\C_1(n)$ and its coset $\C_w(n)$
have the same distance properties, and hence, have the same SNR threshold under
ML decoding. Based on the threshold behavior of good codes \cite{rpe-good-IT}
and the condition $\lambda\ge \lambda^{\star}$, we have $\lim_{n\rightarrow
\infty} P(\C_w,n)=0.$ Moveover, Fano's inequality implies that
\begin{align}
\lim_{n\rightarrow \infty} H(\Xv| W,\Zv)/n&\le \lim_{n\rightarrow \infty}
[1/n+P(\C_w,n)R_1]= 0. \label{eq:fano}
\end{align}
Combining (\ref{eq:equiv}) and (\ref{eq:fano}), we have the desired result.

\subsection{Proof of Theorem~\ref{th:awgn-2}} \label{sec:pft2}

To develop the achievable rate-equivocation pair, we consider an equivalent
channel model illustrated in Fig.~\ref{fig:eqwt}. We observe that the
equivalent channel embeds a binary erasure wiretap channel $X\rightarrow
(Y,Z')$, where $Z'$ is the BEC output with alphabet $\{1,0,-1\}$. The proof can
be outlined as follows.

We first construct a BEC-WT($\epsilon$) and an associated channel with
transition probabilities $f_{Z|Z'}$ so that the channel $X\rightarrow Z$ is
equivalent to the original BI-AWGN with SNR $\lambda$.
To this end, we choose the erasure rate $\epsilon$
as follows
\begin{align*}
\epsilon&=\int_{-\infty}^{\infty}\min\bigl[g(z|X=-1),\, g(z|X=1)\bigr]\, dz 
=2Q\bigl(\sqrt{2\lambda}\bigr).
\end{align*}
Let us define transition probabilities $f_{Z|Z'}$ as
\begin{align}
f(z|Z'=1)&=\left\{
            \begin{array}{ll}
              \frac{g(z|X=1)-g(z|X=-1)}{1-\epsilon} & z\ge 0 \\
              0 & z< 0
            \end{array}
          \right.\notag \\
f(z|Z'=0)&=\left\{
            \begin{array}{ll}
            g(z|X=-1)/\epsilon & z \ge 0\\
            g(z|X=1)/\epsilon & z<0
            \end{array}
          \right.\notag \\
f(z|Z'=-1)&=\left\{
            \begin{array}{ll}
              0 & z\ge0\\
            \frac{g(z|X=-1)-g(z|X=1)}{1-\epsilon} & z<0 .
            \end{array}
          \right. \label{eq:zzf}
\end{align}
We can easily verify that $\sum_{z'}p(z'|x)f(z|z')=g(z|x).$ This implies that
the designed concatenated channel is equivalent to the original
AWGN-WT($\lambda$).
\begin{figure}[hbt]
\centerline{\includegraphics[width=0.9\linewidth,draft=false]{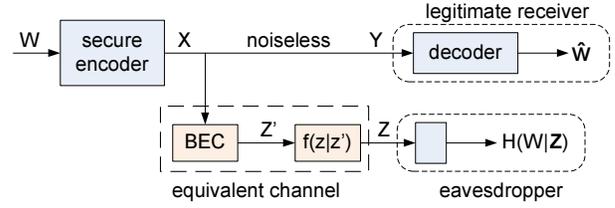}}
  \caption{\small Equivalent type II AWGN wiretap channel}
  \label{fig:eqwt}
\end{figure}

Next, we design secure nested codes for the \emph{upgraded} BEC-WT($\epsilon$).
Note that the confidential message $W$, the BEC output $Z'$, and the received
signal at the eavesdropper $Z$ satisfy the Markov chain $W \rightarrow Z'
\rightarrow Z$. The data processing inequality \cite{Cover} implies that the
normalized equivocation can be bounded as $$H(W|\Zv)/n \ge H(W|\Zv')/n.$$
Finally, we have the desired result by applying Lemma~\ref{lem:bec}.

\bibliographystyle{IEEEtran}
\bibliography{secrecy}

\end{document}

%% file: symboldef.tex
\newtheorem{theorem}{Theorem}
\newtheorem{corollary}{Corollary}
\newtheorem{lemma}{Lemma}

\newtheorem{example}{Example}
\newtheorem{definition}{Definition}

\providecommand{\C}{{\mathcal C}}

\providecommand{\xv}{\mathbf{x}} \providecommand{\yv}{\mathbf{y}}

 \providecommand{\Xv}{\mathbf{X}}
 
\providecommand{\Zv}{\mathbf{Z}} \providecommand{\zv}{\mathbf{z}}

\providecommand{\Wc}{{\mathcal W}}

 \providecommand{\Rs}{{\mathbb R}}